\documentclass[11pt,preprint]{iacrtrans}

\usepackage[colorlinks,linkcolor=red,anchorcolor=green,citecolor=blue]{hyperref}
\usepackage{amsmath,amssymb,booktabs,siunitx}
\allowdisplaybreaks[4]
\author{Mengce Zheng}
\institute{
  College of Information and Intelligence Engineering, Zhejiang Wanli University, China \\ 
  \email{mengce.zheng@gmail.com}, \email{mczheng@zwu.edu.cn}
}

\title{Notes on Small Private Key Attacks on Common Prime RSA\footnote{This work was supported by the National Natural Science Foundation of China, grant number 62002335, Ningbo Natural Science Foundation, grant number 2021J174, and Ningbo Young Science and Technology Talent Cultivation Program, grant number 2023QL007.}
}

\begin{document}

\maketitle


\begin{abstract}
We point out critical deficiencies in lattice-based cryptanalysis of common prime RSA presented in ``Remarks on the cryptanalysis of common prime RSA for IoT constrained low power devices'' [Information Sciences, 538 (2020) 54--68]. 
To rectify these flaws, we carefully scrutinize the relevant parameters involved in the analysis during solving a specific trivariate integer polynomial equation.
Additionally, we offer a synthesized attack illustration of small private key attacks on common prime RSA.
\end{abstract}


\section{Introduction}\label{sect:introduction}

Common prime RSA, i.e., an enhanced RSA \cite{CACM:RivShaAdl78} variant, was first mentioned by Wiener \cite{TIT:Wiener90}, and later refined and named by Hinek \cite{CTRSA:Hinek06}. 
This RSA variant involves two balanced primes $p$ and $q$ with a special structure that provides resistance against previous attacks. 
Hinek defines $p=2ga+1$ and $q=2gb+1$, where $a$ and $b$ are coprime positive integers, and $g$ is a prime. 
Besides, $h=2gab+a+b$ is ensured to be a prime, and hence $(pq-1)/2$ equaling to $gh$ is a semiprime.

Its public/private exponents $e,\ d$ are defined in the key equation $ed\equiv 1 \pmod{\mathrm{lcm}(p-1,q-1)}$.
As $\mathrm{lcm}(p-1,q-1)=\mathrm{lcm}(2ga,2gb)=2gab$, we have
\begin{equation}\label{eqn:key}
	ed\equiv 1 \pmod{2gab},
\end{equation}
for an unknown integer $k$ relatively prime to $2g$. 
We denote the greatest common divisor as $g\simeq N^\gamma$, and its private exponent as $d\simeq N^\delta$. 
We have $0<\gamma<1/2$ for balanced primes, and hence $e$ is approximately $2gab$, implying $e\simeq N^{1-\gamma}$.

Cryptanalysis of common prime RSA has been extensively conducted by various previous works \cite{CTRSA:Hinek06,AC:JocMay06,DCC:SarMai13,AC:LZPL15,IS:MumtazL20}, focusing mainly on polynomial-time small private key attacks. 
We briefly summarize the results of attack bounds on $\delta$ as follows.
\begin{description}
	\item[Wiener's Attack.] 
	Wiener \cite{TIT:Wiener90} used a continued fraction attack to prove that given a public key $(N, e)$ with $\delta<1/4-\gamma/2$, one can factorize the common prime RSA modulus $N$ in polynomial time.
	\item[Hinek's Attack.] 
	Hinek \cite{CTRSA:Hinek06} conducted a systematical study on common prime RSA with two lattice-based attacks. To be concrete, $N$ can be factorized in polynomial time when $\delta<\gamma^2$ or $\delta<2\gamma/5$.
	\item[Jochemsz-May's Attack.] 
	Jochemsz and May \cite{AC:JocMay06} reevaluated the equation introduce by Hinek \cite{CTRSA:Hinek06} and solved it using different unknown variables. The bound on $\delta$ has been further improved to
	\[
	\delta<\frac{1}{4}\left(4+4 \gamma-\sqrt{13+20 \gamma+4 \gamma^2}\right).
	\]
	\item[Sarkar-Maitra's Attack.] 
	Sarkar and Maitra \cite{DCC:SarMai13} showed two improved lattice-based attacks. One is applicable for $\gamma\leq 0.051$ under a complicated condition, and another is applicable for $0.051<\gamma\leq 0.2087$ with the following bound.
	\[
	\delta<\frac{1}{4}-\frac{\gamma}{2}+\frac{\gamma^2}{2}.
	\] 
	\item[Lu et al.'s Attack.] Lu et al.\ \cite{AC:LZPL15} further analyzed the security of common prime RSA by solving simultaneous modular equations and obtained an improvement for $\gamma\geq 0.3872$. The bound on $\delta$ has been further improved to
	\[
	\delta<4\gamma^3,\ \ \gamma>1/4.
	\]
	\item[Mumtaz-Luo's Attack.] Mumtaz and Luo \cite{IS:MumtazL20} applied solving multivariate polynomial equations using a generalized lattice-based method for small private key attack on common prime RSA. They proposed an attack that works when 
	\[
	\delta<2-\gamma-\frac{1}{4} \sqrt{4 \gamma^2-28 \gamma+37}.
	\]
\end{description}

It is worth noting that Mumtaz-Luo's attack utilized incorrect or incomplete parameters, rendering their findings inapplicable. 
Our examination reveals the presence of incorrect parameters and incomplete conditions, prompting us to make the necessary corrections to their approach. 
Consequently, we present a refined security assessment of common prime RSA based on small private key attacks, offering a detailed illustration of its insecure and secure boundaries.

In this work, we adopt lattice-based integer polynomial solving strategy \cite{JC:Coppersmith97}, a technique commonly employed in previous cryptanalysis. 
We carefully check the relevant parameters associated with the specific solving condition in Mumtaz-Luo's attack and identify instances of inappropriate usage and missing explanations. 
Subsequently, we propose corrective measures, leading to the refinement of our small private key attack  through a discussion of an optimizing parameter. 
Our refined attack is effective for the following bound on $\delta$.
\[
\left\{
\begin{aligned}
	\delta &< \gamma+1-\frac{\sqrt{4 \gamma^2+20 \gamma+13}}{4}, & 0<\gamma\leq \frac{3}{10},\\
	\delta &< \frac{4\gamma+1}{11}, & \frac{3}{10}<\gamma<\frac{1}{2}.
\end{aligned}
\right.
\] 
Taking into account previous small private key attacks as well as our refined one, we present an illustrative security assessment of common prime RSA in Figure~\ref{fig:attacks}.
\begin{figure*}[ht]
	\centering
	\includegraphics[width=\textwidth]{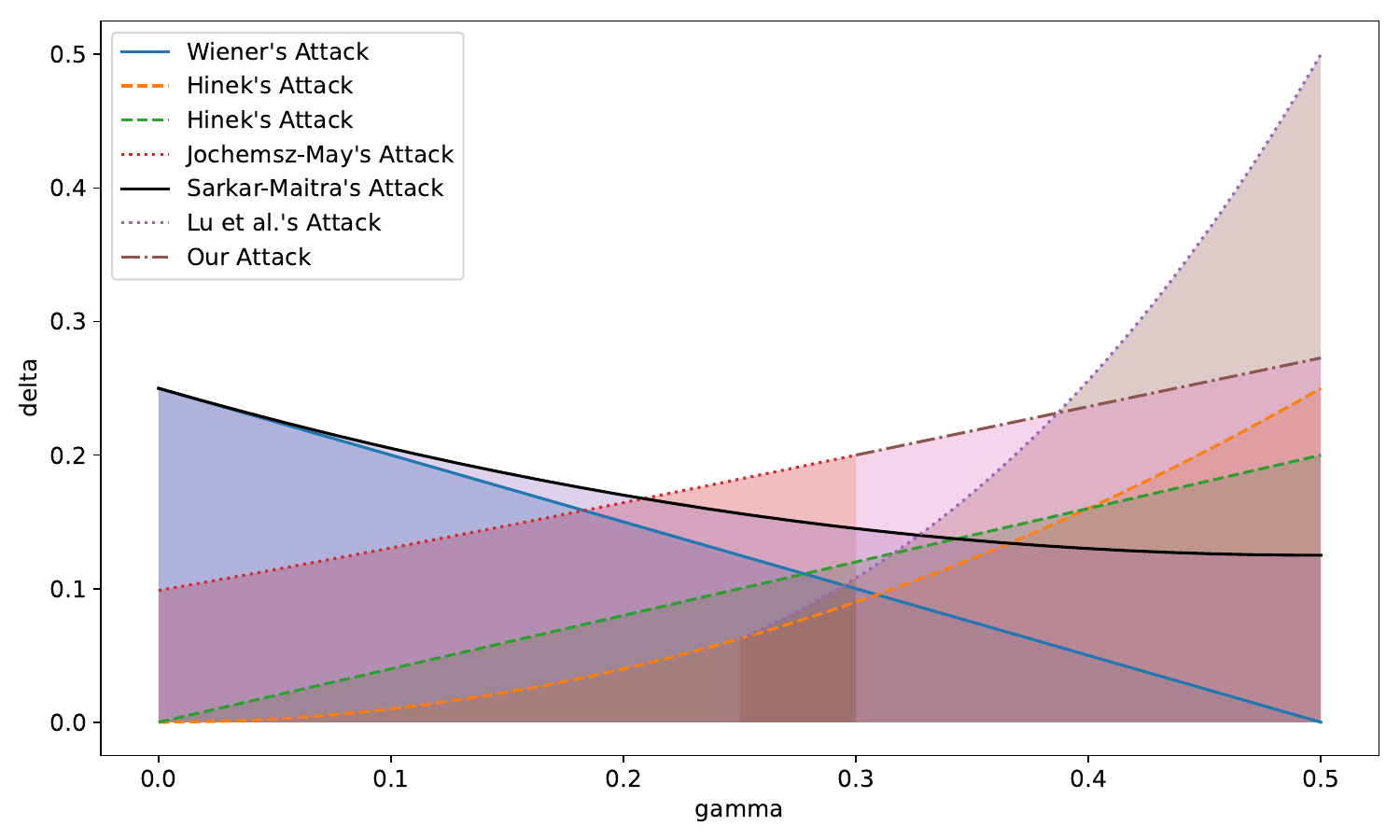}
	\caption{The shadows delineate the attack region on common prime RSA. These attack curves function as critical boundaries differentiating between secure and insecure common prime RSA settings.}
	\label{fig:attacks}
\end{figure*}

The rest of this paper is structured as follows. 
In Section~\ref{sect:preliminary}, we provide an introduction to lattice-based solving method, along with a solving condition for multivariate polynomial equations. 
We review Mumtaz-Luo's attack to point out several existing flaws, and present our corrections with a refined small private key attack on common prime RSA in Section~\ref{sect:review} and Section~\ref{sect:correction}.
We validate our proposed corrected attack through intensive numerical experiments in Section~\ref{sect:experiment}. 
Finally, we conclude the paper in Section~\ref{sect:conclusion}.

\section{Preliminaries}\label{sect:preliminary}

The fundamental concepts include the lattice reduction algorithm, notably the LLL algorithm by Lenstra, Lenstra, and Lovasz \cite{MA:LenLenLov82}, and Coppersmith's lattice-based method \cite{JC:Coppersmith97}, which was later refined as Howgrave-Graham's lemma \cite{IMA:HowgraveGraham97}. 
Additionally, a solving condition essential for finding the root of integer polynomials is introduced. 
For a more comprehensive understanding, interested readers can refer to \cite{PHD:May2003,May10}.

Let us begin by defining lattice $\mathcal{L}$ as the set of all integer linear combinations of linearly independent vectors $\vec{b}_1,\ldots,\vec{b}_{\omega}$. 
In other words, it can be expressed as $\mathcal{L}=\{\sum_{i=1}^{\omega}z_i\vec{b}_i:z_i\in\mathbb{Z}\}$. 
The lattice determinant, denoted as $\det(\mathcal{L})$, is calculated as $\sqrt{\det(BB^{\mathsf{T}})}$, where each $\vec{b}_i$ is considered as a row vector of the basis matrix $B$. 
When dealing with a full-rank lattice with $\omega=n$, the determinant becomes $\det(\mathcal{L})=\lvert \det(B) \rvert$.

The LLL algorithm \cite{MA:LenLenLov82} is a key mathematical tool for efficiently finding approximately short lattice vectors. 
As proven by \cite{PHD:May2003}, the LLL algorithm yields a reduced basis $(\vec{v}_1, \vec{v}_2, \ldots, \vec{v}_{\omega})$ with the following property.
\begin{lemma}\label{lemma:lll}
	The LLL algorithm outputs a reduced basis $(\vec{v}_1,\vec{v}_2,\ldots,\vec{v}_{\omega})$ of a given  $\omega$-dimensional lattice $\mathcal{L}$ satisfying
	\[
	\|\vec{v}_1\|,\ldots,\|\vec{v}_i\|\leq 2^{\frac{\omega(\omega-1)}{4(\omega+1-i)}}\det(\mathcal{L})^{\frac{1}{\omega+1-i}},\quad 1\leq i\leq \omega.
	\]
	Its time complexity is polynomial in $\omega$ and in logarithmic maximal input vector.
\end{lemma}

An essential lemma introduced by Howgrave-Graham \cite{IMA:HowgraveGraham97} provides a principle for determining whether the root of a modular polynomial equation also corresponds to a root over the integers. 
This lemma concerns an integer polynomial $g(x_1,\ldots,x_n)=\sum a_{i_1,\ldots,i_n} x_1^{i_1}\cdots x_n^{i_n}$ and its norm $\|g(x_1,\ldots,x_n)\|:=\sqrt{\sum \lvert a_{i_1,\ldots,i_n} \rvert^2}$.
\begin{lemma}\label{lemma:hg}
	Let $g(x_1,\ldots,x_n)\in\mathbb{Z}[x_1,\ldots,x_n]$ be an integer polynomial, consisting of at most $\omega$ monomials.
	If the two following conditions are satisfied: 
	\begin{itemize}
		\item $g(x_1^{\star},\ldots,x_n^{\star})\equiv 0\pmod{R}$ for $\lvert x_1^{\star}\rvert\leq X_1,$ $\ldots,\lvert x_n^{\star}\rvert \leq X_n$,
		\item $\|g(x_1X_1,\ldots,x_nX_n)\|<R/\sqrt{\omega}$.
	\end{itemize}
	Then $g(x_1^{\star},\ldots,x_n^{\star})=0$ holds over the integers.
\end{lemma}

Combining the LLL algorithm's outputs with Howgrave-Graham's lemma, we can efficiently solve modular/integer polynomial equations. 
Suppose that we have calculated the first $\ell$ many reduced vectors, the key to success lies in satisfying the condition $2^{\frac{\omega(\omega-1)}{4(\omega+1-\ell)}}\det(\mathcal{L})^{\frac{1}{\omega+1-\ell}}<R/\sqrt{\omega}$.
It reduces to $\det(\mathcal{L})<R^{\omega+1-\ell}2^{-\frac{\omega(\omega-1)}{4}}\omega^{-\frac{\omega+1-\ell}{2}}$.
We always have $\ell < \omega \ll R$ and hence it further leads to $\det(\mathcal{L}) < R^{\omega-\epsilon}$ with a tiny error term $\epsilon$. 
We finally derive the following asymptotic solving condition as  
\begin{equation}\label{condition:def}
	\det(\mathcal{L})<R^{\omega},
\end{equation}
which allows us to effectively solve given modular/integer polynomial equations.

The lattice-based solving strategy involves the following stages.
Initially, we generate a set of shift polynomials using the provided polynomial $f(x_1,\ldots,x_n)$ and upper bounds $X_1,\ldots,X_n$. 
These shift polynomials are specifically designed to share a common root modulo $R$. 
Subsequently, we create a lattice by converting the coefficient vectors of each shift polynomial $g_i(x_1X_1,\ldots,x_nX_n)$ into row vectors of a lattice basis matrix. 
Utilizing the LLL algorithm, we then obtain first few reduced vectors. 
These vectors are further transformed into integer polynomials $h_i(x_1,\ldots,x_n)$. 
Once we ensure that the resulting integer polynomials $h_i(x_1,\ldots,x_n)$ are algebraically independent,  the equation system can be effectively solved using trivial methods, thus extracting the desired root.

Several studies have focused on constructing an elegant lattice basis matrix with optimized solving conditions, including works such as \cite{EC:BM05,AC:JocMay06,ACISP:TakayasuK13,AC:LZPL15}.
In this paper, we adopt Jochemsz-May's strategy \cite{AC:JocMay06}, which involves creating a triangular basis matrix, where $\det(\mathcal{L})$ is easy to compute as multiplication of matrix's diagonal elements.
For a comprehensive and detailed explanation, refer to Section~\ref{sect:correction}.

To find the roots of a given trivariate integer polynomial in the specific form as
\[
f(x_1,x_2,x_3)=a_0+a_1 x_1+a_2 x_1^2+a_3 x_2+a_4 x_3+a_5 x_1 x_2+a_6 x_1 x_3+a_7 x_2 x_3,
\]
we should establish upper bounds $X_1$, $X_2$, and $X_3$ for the unknown variables $x_1$, $x_2$, and $x_3$. 
Additionally, $X_{\infty}$ is defined as the maximal individual term value related to the trivariate polynomial, which is given by
\begin{equation}\label{Xinfinity:def}
	X_{\infty} = \|f(x_1X_1,x_2X_2,x_3X_3)\|_\infty.
\end{equation}
To proceed, we introduce the parameter $R=X_{\infty} X_1^{2(s-1)+t} (X_2X_3)^{s-1}$, where $s$ is a positive integer, and $t=\tau s$ with $\tau\geq 0$ to be determined and optimized during the subsequent lattice construction.

We then construct a basis matrix using coefficient vectors of shift polynomials using two monomial sets $\mathcal{S}$ and $\mathcal{M}$. 
We define $f'(x_1,x_2,x_3)=a_0^{-1}f(x_1,x_2,x_3)\bmod{R}$ to set the constant term as $1$. 
The corresponding shift polynomials are given by
\begin{align*}
	& x_1^{i_1}x_2^{i_2}x_3^{i_3}f'(x_1,x_2,x_3)X_1^{2(s-1)+t-i_1} X_2^{s-1-i_2} X_3^{s-1-i_3}, & x_1^{i_1}x_2^{i_2}x_3^{i_3} \in \mathcal{S}, \\
	& Rx_1^{i_1}x_2^{i_2}x_3^{i_3}, & x_1^{i_1}x_2^{i_2}x_3^{i_3} \in \mathcal{M}\setminus\mathcal{S}.
\end{align*}

Upon straightforward and meticulous computations with the above parameters, we establish a parameterized solving condition. 
We obtain the first two vectors of the reduced basis under the proposed procedure and transform them into two polynomials $f_1(x_1,x_2,x_3)$ and $f_2(x_1,x_2,x_3)$ that share a common root over the integers. 
The extraction of the common root can be accomplished using resultant computation or  Gr{\"o}bner basis computation \cite{Grobner}. 
The running time primarily depends on computing the reduced lattice basis and recovering the desired root, both of which can be efficiently achieved in polynomial time with respect to the inputs.

The lattice-based solving strategy is a heuristic approach, as there is no assurance that the derived integer polynomials will always be algebraically independent. 
However, in the realm of lattice-based attacks, it is commonly assumed that the polynomials obtained through the LLL algorithm possess algebraic independence.

\section{Reviewing Mumtaz-Luo's Attack}\label{sect:review}

We begin by reviewing Mumtaz-Luo's analysis of common prime RSA and identify several incomplete or incorrect derivations.
It capitalizes on special property of common prime RSA, specifically, $ed=2gabk+1$ for an unknown integer $k$ from the key equation \eqref{eqn:key}. 
By substituting $2ga=p-1$ and $2gb=q-1$ into $ed=2gabk+1$, Mumtaz and Luo get
\begin{equation}\label{eqn:akbk}
	ed = (p-1)bk+1,\quad ed = (q-1)ak+1.
\end{equation}
Furthermore, it leads to
\[
ed-1+bk = pbk,\quad ed-1+ak = qak.
\]
Multiplying them together yields
\[
(ed-1+bk)(ed-1+ak)=pbk\cdot qak=abk^2N,
\]
which simplifies to
\begin{equation}\label{eqn:four}
	e^2 d^2+ed\left(ak+bk-2\right)-abk^2(N-1)-ak-bk+1=0.
\end{equation}

They merge the variable $k$ with $a$ and $b$, resulting in the following trivariate integer polynomial.
\begin{equation}\label{eqn:three}
	f(x, y, z)=e^2 x^2+e x(y+z-2)-(y+z-1)-(N-1) y z.
\end{equation}
Thus, it turns to finding the root $(x^{\star},y^{\star},z^{\star})=(d, ak, bk)$ of this trivariate integer polynomial.
The estimated upper bound values are 
\begin{equation}\label{eqn:upperbound}
	|x^{\star}|\leq X=N^{\delta},\ |y^{\star}|\leq Y=N^{\delta-\gamma+1/2},\ |z^{\star}|\leq Z=N^{\delta-\gamma+1/2}
\end{equation}
considering that $k\simeq N^\delta$ and $a,b\simeq N^{1/2-\gamma}$.

They employ generalized Coron's reformulation \cite{EC:Coron04,C:Coron07}, similar to Jochemsz-May's strategy \cite{AC:JocMay06}. 
The maximal coefficient $W=\|f(xX,yY,zZ)\|_\infty$ is defined for $f(x,y,z)$, but the specific value of $W$ is not explicitly mentioned.
Through generalized Coron's reformulation, they derive the following solving condition.
\begin{equation}\label{eqn:condition}
	X^{7+9 \tau+3 \tau^2}(Y Z)^{5+9 \tau / 2}<W^{3+3 \tau},
\end{equation} 
which coincides with the solving condition presented in \cite{AC:JocMay06}.
Subsequently, they use $X, Y, Z$, and $W$ and simplify to obtain an inequality involving $\delta$, $\gamma$, and an optimizing parameter $\tau$.
\begin{equation}\label{eqn:wrong}
	3 \delta \tau^2+(12 \delta+3 \gamma-9 / 2) \tau+11 \delta+4 \gamma-5<0.
\end{equation}
To maximize its left side, they let
\begin{equation}\label{eqn:tau}
	\tau=\frac{3-2\gamma-8\delta}{4\delta}.
\end{equation}
Substituting $\tau$ back into \eqref{eqn:wrong}, they obtain the inequality with respect to $\delta$ and $\gamma$.
\begin{equation*}
	\frac{-16 \delta^2-32 \delta \gamma+64 \delta-12 \gamma^2+36 \gamma-27}{16 \delta}<0 .
\end{equation*}
Thus, they obtain the following bound, omitting the tiny term $\epsilon$ presented in the original bound \cite[Formula~(12)]{IS:MumtazL20}.
\begin{equation}\label{eqn:delta1}
	\delta<2-\gamma-\frac{1}{4} \sqrt{4 \gamma^2-28 \gamma+37}.
\end{equation}

Mumtaz and Luo recover the root $(x^{\star},y^{\star},z^{\star})=(d, ak, bk)$ through the resultant computation. 
Thus, they claim that the factorization of common prime modulus is done in polynomial time.
However, their derivation from \eqref{eqn:condition} to \eqref{eqn:wrong} is not smooth and intuitive, as $W$ is not explicitly given during their analysis. 
Conversely, we aim to discover the value of $W$ they use through the following inverse computation. 
Assuming that the condition \eqref{eqn:wrong} is correctly derived for $X=N^{\delta}, Y=Z=N^{\delta-\gamma+1/2}$ given in \eqref{eqn:upperbound} and $W=N^{\xi}$ with a fixed $\xi$ to be recovered, we have
\[
N^{(7+9 \tau+3 \tau^2)\delta}N^{2(5+9 \tau / 2)(\delta-\gamma+1/2)}<N^{(3+3 \tau)\xi}.
\]
Simplifying it gives us
\[
(7+9 \tau+3 \tau^2)\delta+(5+9 \tau / 2)(2\delta-2\gamma+1)<(3+3 \tau)\xi.
\]
After rearrangement, we have
\begin{equation}\label{eqn:right}
	3 \delta \tau^2+(18 \delta-9 \gamma-3\xi+9 / 2) \tau+17 \delta-10 \gamma-3\xi+5<0.
\end{equation}
Comparing the corresponding coefficients in \eqref{eqn:wrong} and \eqref{eqn:right}, we must ensure 
\[
\left\{
\begin{aligned}
	18 \delta-9 \gamma-3\xi+9 / 2 & =12 \delta+3 \gamma-9 / 2,\\
	17 \delta-10 \gamma-3\xi+5 & =11 \delta+4 \gamma-5.
\end{aligned}
\right.
\]
Solving the above simultaneous equations, we encounter a contradiction about $\xi$ as follows.
\[
\left\{
\begin{aligned}
	\xi & =2\delta-4\gamma+3,\\
	\xi & =2\delta-14\gamma/3+10/3.
\end{aligned}
\right.
\]

Even if we still assume that condition \eqref{eqn:wrong} is correct, the derived bound on $\delta$ is not accurate. 
Mumtaz-Luo's bound on $\delta$ is presented as
\[
\delta<2-\gamma-\frac{1}{4} \sqrt{4 \gamma^2-28 \gamma+37}.
\]
However, it overestimates the capability of the small private key attack since they ignore a crucial prerequisite, i.e., $\tau\geq 0$, used in the lattice-based method. 
The optimizing parameter is set to $\tau=(3-2\gamma-8\delta)/(4\delta)$ according to \eqref{eqn:tau}. 
Hence, we must ensure that
\[
\frac{3-2\gamma-8\delta}{4\delta}\geq 0,
\]
which results in $3-2\gamma-8\delta\geq 0$.
Therefore, we obtain another constrained bound on $\delta$.
\begin{equation}\label{eqn:delta2}
	\delta\leq \frac{3-2\gamma}{8}.
\end{equation}

Taking both \eqref{eqn:delta1} and \eqref{eqn:delta2} into consideration, we obtain an accurate bound on $\delta$ as follows.
\[
\delta < 2-\gamma-\frac{1}{4} \sqrt{4 \gamma^2-28 \gamma+37},\quad \mathrm{and} \quad \delta \leq \frac{3-2\gamma}{8}.
\]
This leads to  
\[
\delta\leq \frac{3-2\gamma}{8}.
\]
However, this result is actually a weaker bound on $\delta$, as Mumtaz and Luo pick $\tau=(3-2\gamma-8\delta)/(4\delta)$, which can be further optimized. 

To provide a better attack, we calculate a stronger bound on $\delta$ using another simplified approach. 
It can be inferred from \eqref{eqn:wrong} that
\[
(3\tau^2+12\tau+11)\delta+3\gamma\tau-9\tau/2+4\gamma-5<0.
\]
Therefore, as $\tau\geq 0$, we have 
\[
\delta<\frac{9\tau/2-3\gamma\tau-4\gamma+5}{3\tau^2+12\tau+11}.
\]
Let $\lambda(\tau):={(9\tau/2-3\gamma\tau-4\gamma+5)}/{(3\tau^2+12\tau+11)}$ for a given $0<\gamma<1/2$.
Its derivative with respect to $\tau$ is
\[
\frac{\partial\, \lambda(\tau)}{\partial\, \tau} = \frac{3 \left((6\gamma- 9) \tau^2 +  (16\gamma- 20) \tau +10\gamma-7\right)}{2 \left(3 \tau^2+ 12 \tau + 11\right)^2}.
\]

Because $0<\gamma<1/2$, we have $6\gamma- 9<0$, $16\gamma- 20<0$, $10\gamma-7<0$ and hence the numerator of the above derivative is negative.
Thus, $\lambda(\tau)$ is decreasing in the domain of $\tau\geq 0$, and its maximum is taken at $\tau = 0$.
Therefore, we derive the bound $(5-4\gamma)/11$ when maximizing $\lambda(\tau)$ by setting $\tau=0$.
Through a direct comparison, it can be easily checked that $(5-4\gamma)/11 > (3-2\gamma)/8$.
Thus, under Mumtaz-Luo's analysis idea, the exact bound on $\delta$ is 
\[
\delta<\frac{5-4\gamma}{11}.
\]

To conclude, Mumtaz-Luo's attack \cite{IS:MumtazL20} suffers from several fatal flaws, rendering their attack result incorrect and inapplicable. 
The summary of these flaws is as follows.
\begin{description}
	\item[Repetitive Approach.] 
	They use generalized Coron's reformulation to solve a trivariate integer polynomial, but the specific monomial form and its relevant unknown variables are the same as those analyzed in Jochemsz-May's attack \cite[Section~5.2]{AC:JocMay06}. 
	\item[Incorrect Parameter $W$.] 
	The parameter $W$, referring to the maximal coefficient of $f(xX,yY,zZ)$ with upper bounds $X,Y,Z$, is not explicitly provided. Moreover, the value of $W$ implied by their analysis derivation contradicts itself, leading to uncertainty and inconsistency in their analysis.
	\item[Incorrect Bound on $\delta$.] 
	They ignore a crucial prerequisite described in the lattice-based method, where the optimizing parameter $\tau$ should satisfy $\tau\geq 0$. This oversight results in an incorrect bound on $\delta$, and the theoretical bound values given in Theorem~1 do not agree with the ones presented in Table~2.
	\item[Incomplete Factorization.] 
	Their analysis lacks how to factorize the given common prime RSA modulus if the root $(d,ak,bk)$ is recovered. This omission renders the proof of Theorem~1 incomplete and leaves their attack without a crucial step.
\end{description}
These mentioned flaws in Mumtaz-Luo's analysis raise significant concerns about the reliability and validity of their proposed attack. 
Addressing these issues is essential before considering the effectiveness of their approach in practical attack scenarios.

\section{Corrections to Previous Cryptanalysis}\label{sect:correction}

In light of the flaws in Mumtaz-Luo's analysis \cite{IS:MumtazL20}, we present a refined small private key attack and make corrections to their previous findings.
\begin{proposition}\label{prop:main}
	Given a common prime RSA modulus $N=pq$ for balanced primes $p,q$ and $(p-1)/2,\ (q-1)/2$ having a prime $g\simeq N^{\gamma}$, and $e\simeq N^{1-\gamma},\ d\simeq N^{\delta}$ satisfying $ed\equiv 1 \bmod{\mathrm{lcm}(p-1,q-1)}$, then one can factorize the given common prime RSA modulus in polynomial time if
	\begin{equation}\label{delta:main}
		\left\{
		\begin{aligned}
			\delta &< \gamma+1-\frac{\sqrt{4 \gamma^2+20 \gamma+13}}{4}, & 0<\gamma\leq \frac{3}{10},\\
			\delta &< \frac{4\gamma+1}{11}, & \frac{3}{10}<\gamma<\frac{1}{2}.
		\end{aligned}
		\right.
	\end{equation}
\end{proposition}

\begin{proof}
	To refine the small private key attack, we focus on solving the trivariate integer polynomial
	\[
	f(x_1, x_2, x_3)=e^2 x_1^2+e x_1(x_2+x_3-2)-(x_2+x_3-1)-(N-1) x_2 x_3.
	\]
	Following the lattice-based solving strategy \cite{AC:JocMay06}, it can be rewritten as
	\begin{equation}\label{eqn:solve}
		f(x_1, x_2, x_3)=1-2e x_1+e^2 x_1^2-x_2-x_3+e x_1x_2+e x_1x_3+(1-N) x_2 x_3
	\end{equation}
	for $a_0=1,a_1=-2e,a_2=e^2,a_3=-1,a_4=-1,a_5=e,a_6=e$ and $a_7=1-N$.
	We aim to find the desired root $(x_1^{\star}, x_2^{\star}, x_3^{\star})=(d, ak, bk)$.
	
	Using $e\simeq N^{1-\gamma}$ and $d\simeq N^{\delta}$, we estimate $ak\simeq N^{\delta-\gamma+1/2}$ and $bk\simeq N^{\delta-\gamma+1/2}$ based on \eqref{eqn:akbk}.
	The upper bounds $X_i$ are identical to the previous ones \eqref{eqn:upperbound}.  
	\[
	X_1=N^{\delta},\ X_2=N^{\delta-\gamma+1/2},\ X_3=N^{\delta-\gamma+1/2}.
	\]
	The maximal term $X_{\infty}$ is used in the solving condition, and it can be computed as 
	\[
	\begin{aligned}
		X_{\infty} 
		& =\|f(x_1X_1,x_2X_2,x_2X_2)\|_\infty \\
		&= \max\left\{ \lvert a_0\rvert, \lvert a_1\rvert X_1, \lvert a_2\rvert X_1^2, \lvert a_3\rvert X_2, \lvert a_4\rvert X_3, \lvert a_5\rvert X_1 X_2, \lvert a_6\rvert X_1 X_3, \lvert a_7\rvert X_2 X_3 \right\} \\
		& = \max\left\{ N^{\delta-\gamma+1}, N^{2\delta-2\gamma+2}, N^{\delta-\gamma+1/2}, N^{2\delta-2\gamma+3/2}, N^{2\delta-2\gamma+2} \right\} \\
		& = N^{2\delta-2\gamma+2}.
	\end{aligned}
	\]
	
	We use one extra shift of $x_1$ in the lattice-based method as it is smaller than $x_2$ and $x_3$.
	We construct two monomial sets $\mathcal{S}$ and $\mathcal{M}$ for $s$ and $t=\tau s$ as follows.
	\begin{align*}
		\mathcal{S} = \bigcup_{\substack{0\leq j\leq t}}
		\big\{ & x_1^{i_1+j}x_2^{i_2}x_3^{i_3}:
		x_1^{i_1}x_2^{i_2}x_3^{i_3} \in f^{s-1}\big\},\\
		\mathcal{M} = \bigcup_{\substack{0\leq j\leq t}}
		\big\{ & x_1^{i_1+j}x_2^{i_2}x_3^{i_3}:
		x_1^{i_1}x_2^{i_2}x_3^{i_3} \in f^{s}\big\}.
	\end{align*}
	We know the relationship between monomials $x_1^{i_1}x_2^{i_2}x_3^{i_3}$ in $\mathcal{S}$ and $\mathcal{M}$ and the corresponding exponents $i_1,i_2,i_3$ via the expansion of $f^{s-1}$ and $f^s$.
	\begin{align*}
		x_1^{i_1}x_2^{i_2}x_3^{i_3}\in \mathcal{S} & \Leftrightarrow
		\left\{
		\begin{aligned}
			i_2 & = 0,\ldots,s-1,\\
			i_3 & = 0,\ldots,s-1,\\
			i_1 & = 0,\ldots,2(s-1)-i_2-i_3+t.
		\end{aligned}
		\right.\\
		x_1^{i_1}x_2^{i_2}x_3^{i_3}\in \mathcal{M} & \Leftrightarrow
		\left\{
		\begin{aligned}
			i_2 & = 0,\ldots,s,\\
			i_3 & = 0,\ldots,s,\\
			i_1 & = 0,\ldots,2s-i_2-i_3+t.
		\end{aligned}
		\right.
	\end{align*}
	
	The constant term of $f(x_1,x_2,x_3)$ is required to be $1$ and fortunately $a_0$ is exactly $1$.
	Thus, We define the shift polynomials $g_{[i_1,i_2,i_3]}$ according to distinct monomial $x_1^{i_1}x_2^{i_2}x_3^{i_3}$ in $\mathcal{S}$ and $\mathcal{M}$ for $R=X_{\infty} X_1^{2(s-1)+t} (X_2X_3)^{s-1}$ as follows.
	\begin{align*}
		& x_1^{i_1}x_2^{i_2}x_3^{i_3}f(x_1,x_2,x_3)X_1^{2(s-1)+t-i_1} X_2^{s-1-i_2} X_3^{s-1-i_3},\quad & x_1^{i_1}x_2^{i_2}x_3^{i_3}\in \mathcal{S},\\
		& R x_1^{i_1}x_2^{i_2}x_3^{i_3} ,\quad & x_1^{i_1}x_2^{i_2}x_3^{i_3}\in \mathcal{M}\setminus \mathcal{S}.
	\end{align*}
	
	The coefficient vectors of $g_{[i_1,i_2,i_3]}$, where $x_iX_i$ is substituted for each $x_i$, generate a lattice $\mathcal{L}$.
	We need to compute the lattice determinant $\det(\mathcal{L})$, and 
	the diagonal elements of $g_{[i_1,i_2,i_3]}(x_1X_1,x_2X_2,x_3X_3)$
	are $X_1^{2(s-1)+t} (X_2X_3)^{s-1}=R/X_{\infty}$ and $RX_1^{i_1}X_2^{i_2}X_3^{i_3}$, respectively.
	So, we obtain
	\[
	\left(R/X_{\infty}\right)^{s_0}R^{s_R}X_1^{s_1}X_2^{s_2}X_3^{s_3}<R^{\omega},
	\]
	where $s_0=\lvert \mathcal{S} \rvert$, 
	$s_j=\sum_{x_1^{i_1}x_2^{i_2}x_3^{i_3}\in \mathcal{M}\setminus \mathcal{S}} i_j$,
	$s_R=\lvert \mathcal{M}\setminus \mathcal{S} \rvert$ and $\omega=\lvert S_R \rvert$. 
	Moreover, $\omega=\lvert \mathcal{M} \rvert=\lvert \mathcal{S} \rvert+\lvert \mathcal{M}\setminus \mathcal{S} \rvert=s_0+s_R$.
	Therefore, $R$ in the left and right sides cancel each other out, and we obtain
	\begin{equation}\label{condition:general}
		X_1^{s_1}X_2^{s_2}X_3^{s_3}<X_{\infty}^{s_0},
	\end{equation}
	for $s_j=\sum_{x_1^{i_1}x_2^{i_2}x_3^{i_3}\in \mathcal{M}\setminus \mathcal{S}} i_j$ and $s_0=\lvert \mathcal{S} \rvert$.
	By calculating $s_j$ for $j=0, 1, 2, 3$ through the above deduction, we obtain
	\begin{equation*}
		\begin{aligned}
			s_0 & = \lvert \mathcal{S} \rvert = \sum_{i_2=0}^{s-1}\sum_{i_3=0}^{s-1}\sum_{i_1=0}^{2(s-1)-i_2-i_3+t} 1 = s^3+s^2 t+o(s^3),\\
			s_1 & = \sum_{x_1^{i_1}x_2^{i_2}x_3^{i_3}\in \mathcal{M}\setminus \mathcal{S}} i_1 = \frac{7s^3}{3}+3s^2 t+s t^2+o(s^3),\\
			s_2 & = \sum_{x_1^{i_1}x_2^{i_2}x_3^{i_3}\in \mathcal{M}\setminus \mathcal{S}} i_2 = \frac{5s^3}{3}+\frac{3s^2 t}{2}+o(s^3),\\
			s_3 & = \sum_{x_1^{i_1}x_2^{i_2}x_3^{i_3}\in \mathcal{M}\setminus \mathcal{S}} i_3 = \frac{5s^3}{3}+\frac{3s^2 t}{2}+o(s^3).
		\end{aligned}	
	\end{equation*}
	Using $t=\tau s$ and skipping lower terms $o(s^3)$ gives us
	\begin{equation}\label{sum:general}
		s_0=(1+\tau)s^3,\ s_1=\left(\frac{7}{3}+3\tau+\tau^2\right)s^3,\ s_2=s_3=\left(\frac{5}{3}+\frac{3\tau}{2}\right)s^3.
	\end{equation}
	
	We substitute the values of $X_1,X_2,X_3$, $X_{\infty}$, $s_j$ for $j=0,1,2,3$ into \eqref{condition:general}.
	This results in the inequality
	\[
	\delta\left(\frac{7}{3}+3\tau+\tau^2\right)+2\left(\delta-\gamma+\frac{1}{2}\right)\left(\frac{5}{3}+\frac{3\tau}{2}\right)<(2\delta-2\gamma+2)(1+\tau).
	\]
	Further simplifying, we get
	\[
	\delta\tau^2 + \left(4\delta-\gamma-\frac{1}{2}\right)\tau + \frac{11}{3}\delta-\frac{4\gamma}{3}-\frac{1}{3}<0,
	\]
	which can be reduced to
	\[
	6\delta\tau^2 + \left(24\delta-6\gamma-3\right)\tau + 22\delta-8\gamma-2<0.
	\]
	
	To minimize the left side, we find the optimizing value of $\tau=(2\gamma-8\delta+1)/(4\delta)$.
	Therefore, we get
	\[
	-\frac{3(2\gamma-8\delta+1)^2}{8\delta} + 22\delta-8\gamma-2 < 0.
	\]
	Simplifying it, we obtain
	\[
	16\delta^2-32(\gamma+1)\delta+3(2\gamma+1)^2>0.
	\]
	From this inequality, we can deduce the bound on $\delta$.
	\begin{equation}\label{delta:general}
		\delta<\gamma+1-\frac{\sqrt{4 \gamma^2+20 \gamma+13}}{4}.
	\end{equation}
	Additionally, we should ensure that $\tau=(2\gamma-8\delta+1)/(4\delta)$ meets its prerequisite $\tau\geq 0$, which implies $\delta\leq (2\gamma+1)/8$.
	Comparing it with \eqref{delta:general}, we deduce that \eqref{delta:general} holds for $0<\gamma\leq 3/10$ in our proposed analysis.
	
	Therefore, we determine the bound on $\delta$ for $3/10<\tau<1/2$ and we should set $\tau=0$. 
	It directly indicates
	\[
	11\delta-4\gamma-1<0,
	\] 
	which simplifies to
	\begin{equation}\label{delta:other}
		\delta<\frac{4\gamma+1}{11}.
	\end{equation}
	Gathering \eqref{delta:general} and \eqref{delta:other} together, we finally derive the following result. 
	\[
	\left\{
	\begin{aligned}
		\delta &< \gamma+1-\frac{\sqrt{4 \gamma^2+20 \gamma+13}}{4}, & 0<\gamma\leq \frac{3}{10},\\
		\delta &< \frac{4\gamma+1}{11}, & \frac{3}{10}<\gamma<\frac{1}{2}.
	\end{aligned}
	\right.
	\]
	
	Once we have obtained more integer polynomials apart from $f(x_1,x_2,x_3)$, and they share the common root $(x_1^{\star},x_2^{\star},x_3^{\star})=(d,ak,bk)$ over the integers, we can proceed to extract $d$, $ak$, and $bk$ to factorize $N$.
	To factorize the given common prime RSA modulus $N$ using the obtained values $d$, $ak$, and $bk$, we makes use of the fact that $\gcd(a,b)=1$. 
	So, $k=\gcd(ak,bk)$ is first computed by $k=\gcd(x_2^{\star},x_3^{\star})$.
	Then, we know $a=x_2^{\star}/k$ and $b=x_3^{\star}/k$.
	Next, we apply $ed=2gabk+1$ with known $a$, $b$, and $k$ to compute $g=(ed-1)/(2abk)$.
	Finally, using $a$, $b$, and $g$, we can find $p=2ga+1$ and $q=2gb+1$. 
	These values of $p$ and $q$ give us the factorization of $N$.
	The root extraction and factorization can be done in time that is a polynomial regarding $\log N$ and $s$.
\end{proof}

\begin{remark}
	In addition to addressing the flaws in Mumtaz-Luo's analysis, we have further enhanced Jochemsz-May's attack by refining the bound on $\delta$ for $3/10<\gamma<1/2$. 
	It is important to note that Mumtaz-Luo's analysis essentially reproduces Jochemsz-May's attack, albeit with incorrect and incomplete derivation.
\end{remark}

\section{Validating Experiments}\label{sect:experiment}

The experimental results are provided to demonstrate the performance of the proposed small private key attack based on Proposition~\ref{prop:main}. 
The experiments were carried out on a laptop running a 64-bit Windows~10 operating system with Ubuntu~22.04 installed on WSL~2.
We utilized SageMath \cite{SageMath2023} for conducting the experiments, and the parameters for generating the common prime RSA instances were randomly chosen.

Initially, we generated a common prime RSA modulus $N$ with bit-size denoted by $\ell$ using a predetermined $\gamma$ value. 
Subsequently, we generated the private exponent $d$ with a predetermined bit-size in each experiment. 
To derive the public exponent $e$, we utilized $ed\equiv 1\pmod{\mathrm{lcm}(p-1, q-1)}$. 
Additionally, the bit-size of $d$ was gradually increased to achieve a larger $\delta$ for performing a successful small private key attack.

To execute the proposed attack, we carefully constructed a lattice via suitable integers $s$ and $t$. 
Table~\ref{table:experiment} provided the experimental results for our proposed small private key attack. 
The column $\gamma \ell$ represents the bit-size of $g$ in the generated common prime RSA instances, while $\ell_e$ denotes the bit-size of $e$. 
The column $\delta_t \ell$ provides the theoretical bound on $d$, and the corresponding experimental result is presented in the column $\delta_e \ell$.
The column $\textsf{AR}$ indicates the achieving rate $\delta_e/\delta_t$ of our experimental bound in estimating the distance from the theoretical one.
The lattice settings are controlled by $s$ and $t$, with the lattice dimension specified in the $\omega$ column. 
The time consumption of our experiments is recorded in the \textsf{Time} column, measured in seconds.

\begin{table}[ht]
	\begin{center}
		\begin{minipage}{\textwidth}
			\footnotesize
			\caption{Experimental results of our proposed attack on common prime RSA}\label{table:experiment}
			\begin{tabular*}{\textwidth}{@{\extracolsep{\fill}}cccccccccc@{\extracolsep{\fill}}}
				\toprule
				$\ell$ & $\gamma \ell$ & $\ell_e$ & $\delta_t \ell$ & $\delta_e \ell$ & \textsf{AR} & $s$ & $t$ & $\omega$ & \textsf{Time}          \\ \midrule
				1024       & 205               & 818        & 168                 & 118                 & $70.2\%$    & 2   & 1   & 36       & \SI{1.238}{\second}    \\
				1024       & 256               & 765        & 186                 & 134                 & $72.0\%$    & 2   & 1   & 36       & \SI{1.151}{\second}    \\
				1024       & 307               & 716        & 205                 & 145                 & $70.7\%$    & 2   & 0   & 27       & \SI{0.339}{\second}    \\
				1024       & 359               & 663        & 223                 & 159                 & $71.3\%$    & 2   & 0   & 27       & \SI{0.438}{\second}    \\
				1024       & 410               & 612        & 242                 & 189                 & $78.1\%$    & 3   & 0   & 64       & \SI{20.535}{\second}   \\
				1024       & 461               & 562        & 261                 & 209                 & $80.1\%$    & 4   & 0   & 125      & \SI{1168.909}{\second} \\ \midrule
				2048       & 408               & 1639       & 337                 & 236                 & $70.0\%$    & 2   & 1   & 36       & \SI{3.637}{\second}    \\
				2048       & 512               & 1535       & 373                 & 268                 & $71.8\%$    & 2   & 1   & 36       & \SI{3.295}{\second}    \\
				2048       & 615               & 1432       & 410                 & 295                 & $71.9\%$    & 2   & 1   & 36       & \SI{3.126}{\second}    \\
				2048       & 716               & 1331       & 447                 & 349                 & $78.1\%$    & 3   & 0   & 64       & \SI{81.463}{\second}   \\
				2048       & 820               & 1225       & 484                 & 397                 & $82.0\%$    & 4   & 0   & 125      & \SI{4728.206}{\second} \\
				2048       & 920               & 1123       & 521                 & 445                 & $85.4\%$    & 4   & 0   & 125      & \SI{3816.764}{\second} \\ \bottomrule
			\end{tabular*}
		\end{minipage}
	\end{center}
\end{table}

Throughout each experiment, we collected sufficient integer polynomials that met the solvable requirements after running the LLL algorithm. 
As indicated in Table~\ref{table:experiment}, the running time increases while the lattice dimension becomes higher or the given modulus gets larger. 
After obtaining the integer polynomial equations with a shared root, we successfully recovered $(x_1^{\star},x_2^{\star},x_3^{\star})$ in attacks on generated instances. 
Consequently, we retrieved $d$, $ak$, and $bk$, enabling us to factorize $N$ using the root extraction approach. 
However, the experimental results fell short of reaching the theoretical bound, likely due to limited computing resources. 
It is evident that the attack performance improves with the use of a lattice of increasing dimension.
We believe that practical attack results can be further optimized by constructing lattices with higher dimensions.

\section{Conclusion}\label{sect:conclusion}

We thoroughly examine Mumtaz-Luo's attack on common prime RSA and identify several critical flaws that render their attack inapplicable. 
To rectify these issues, we provide corrections to the previous cryptanalysis and present a refined small private key attack.
The experimental results demonstrate the correctness of our refined attack. 
Notably, our attack successfully penetrates common prime RSA instances that employ small private keys, allowing us to factorize the given modulus in reasonable time consumption.

\bibliographystyle{alpha}
\bibliography{ref}

\end{document}